\documentclass{amsart}
\usepackage{amsmath, color}
\usepackage{amsfonts}
\usepackage{latexsym}
\usepackage{ifpdf}
\usepackage{graphicx,amssymb,lineno}
\usepackage{amssymb}
\usepackage{mathrsfs}
\usepackage{cite}
\usepackage{extarrows}
\usepackage{ulem}
\usepackage{placeins} 
\usepackage{booktabs}
\usepackage{array}
\usepackage{tikz}
\usepackage{tabularx}
\usepackage{booktabs}   
\usepackage{pifont}     
\usepackage{ytableau}
\usepackage{ytableau}
\ytableausetup{
  mathmode,
  boxsize=1.3em,
  smalltableaux
}

\usetikzlibrary{arrows.meta,positioning,calc,decorations.pathreplacing}

\allowdisplaybreaks[4]
\newcommand{\Var}{\text{Var}}

\newtheorem{theorem}{Theorem}[section]
\newtheorem{lemma}[theorem]{Lemma}

\newtheorem{remark}[theorem]{Remark}

\newtheorem{corollary}[theorem]{Corollary}

\theoremstyle{definition}

\newtheorem{example}[theorem]{Example}

\theoremstyle{remark}

\numberwithin{equation}{section} \errorcontextlines=0
\begin{document}

\title[State-independent uncertainty relations]
{Multipartite State-independent Uncertainty Relations}
\author{Yiling Wang}
\address{Department of Mathematics,
   North Carolina State University,
   Raleigh, NC 27695, USA}
\email{ywang327@ncsu.edu}
\author{Kailash Misra}
\address{Department of Mathematics,
   North Carolina State University,
   Raleigh, NC 27695, USA}
\email{misra@ncsu.edu}
\author{Naihuan Jing}
\address{Department of Mathematics,
   North Carolina State University,
   Raleigh, NC 27695, USA}
\email{jing@ncsu.edu}

\thanks{*Corresponding author: Naihuan Jing}

\begin{abstract}
Uncertainty relations quantify fundamental limits on simultaneous measurement of quantum observables. 
While conventional formulations are state-dependent, state-independent uncertainty relations (SIURs) impose universal bounds determined solely by the algebraic structure of the operators, with applications across metrology, quantum cryptography, and entanglement detection. 
Despite extensive study, exact analytical variance-based SIURs have so far been established primarily for one- and bi-partite systems, {
while entropic SIURs have been extended to multipartite and memory-assisted settings through information-theoretic constructions. In contrast, exact variance-based SIURs beyond the bipartite level have remained analytically unresolved.}
Here we develop a representation-theoretic framework for multipartite SIURs in collective spin-$\tfrac{1}{2}$ systems. 
Using the Clebsch--Gordan decomposition 
and extremal analysis of total spin variance, we derive exact state-independent bounds up to quintipartite systems. 
A clear structural dichotomy emerges: odd $n$ systems exhibit strictly positive universal bounds (e.g., $\Delta^2(\mathfrak{su}_2)\!\ge\!4/11$ for $n=3$), whereas even $n$ admit vanishing variance on trivial sectors but retain positive reduced-space bounds (e.g., $\Delta^2(\mathfrak{su}_2)\!\ge\!1/8$ for $n=4$). 
These results establish the first unified, algebraic framework for multipartite variance-based SIURs in qubit ensembles.
\end{abstract}

\maketitle
\section{Introduction}

Uncertainty relations form a cornerstone of quantum mechanics,  formalizing the fundamental limits on the simultaneous measurement 
of incompatible observables. Since Heisenberg's seminal insight~\cite{heisenberg1927} and the rigorous formulations of Kennard, Robertson, and Schr\"odinger~\cite{kennard1927,robertson1929,schrodinger1930}, such relations have provided indispensable constraints for quantum measurement and dynamics. The canonical variance-based inequality,
\begin{equation}
\Delta(A)\Delta(B)\ge \frac{1}{2}\big|\langle [A,B]\rangle\big|,
\end{equation}
expresses this tradeoff in terms of the variances $\Delta^2(A)=\Var(A)=\langle A^2\rangle-\langle A\rangle^2$ and $\Delta(B)$ evaluated on a quantum state $|v\rangle$. {
However, the lower bound depends explicitly on the state through the commutator expectation value and can vanish whenever $\langle [A,B]\rangle_v=0$ or $|v\rangle$ is an eigenstate of either observable. In \cite{maccone2014, coles2017, schwonnek2017} the sum-form was proposed to rectify this situation, subsequently weighted form \cite{XiaoJing16} of uncertainty was also introduced. 
One notices that a state-dependent uncertainty relation (usually variance-based) provides a lower bound that explicitly depends on the particular quantum state under consideration. 
Such bounds may vanish for special states (e.g., eigenstates or states with vanishing commutator expectation values), and therefore quantify measurement incompatibility relative to a given state. 
In contrast, a state-independent uncertainty relation (SIUR) establishes a universal lower bound determined solely by the system or the representation in which they act, and holds uniformly for all states in the Hilbert space. 
State-independent bounds therefore reveal intrinsic, representation-level limitations on simultaneous measurability, independent of state preparation. 
Such universal bounds are particularly relevant in scenarios where the quantum state is unknown, fluctuating, or adversarially chosen, as in entanglement certification, metrological benchmarking, and noise-resilient quantum protocols.}

To eliminate this dependence, state-independent uncertainty relations (SIURs) have been developed~\cite{busch2014,dammeier2015,deGuise2018}. These bounds are not only algebraic, as they depend only on the structure of the operators, and also reveal information about the systems as it remain valid for all states. Such universal formulations provide intrinsic noise floors and have found wide application across quantum metrology~\cite{Giovannetti2006PRL,giovannetti2011,pezze2018,Hall2004PRA,kitagawa1993}, quantum thermodynamics and quantum control~\cite{Deffner2017JPA,Correa2015PRL}, and entanglement detection~\cite{hofmann2003,guhne2004,GuhneToth2009,toth2007,Guhne2002PRA}. 
{
In particular, uncertainty- and variance-based criteria have been extended to 
multipartite and entanglement-depth detection using collective observables and 
quantum Fisher information~\cite{Hyllus2012,Toth2014,Ren2021PRL,Liu2024PLA}, 
as well as to numerical and semidefinite-programming approaches for genuine 
multipartite entanglement certification~\cite{Jungnitsch2011}. 
Recent reviews further clarify the structural aspects of multipartite 
entanglement~\cite{Horodecki2024Multipartite}. }
A simple example is the qubit sum relation
\begin{equation}
\Delta^2(\sigma_x)+\Delta^2(\sigma_y)+\Delta^2(\sigma_z)\ge 2,
\end{equation}
which holds for all single-qubit states and reflects a 
quantum uncertainty for spin-$1/2$ observables.

Two complementary paradigms dominate uncertainty theory. Variance-based formulations follow the Heisenberg--Robertson--Schr\"odinger type, linking measurement fluctuations to the geometry of Hilbert space~\cite{coles2017,maccone2014}. Entropic formulations, pioneered by Maassen and Uffink~\cite{maassen1988} and developed extensively thereafter~\cite{wehner2010,coles2017}, recast uncertainty in information-theoretic terms, bounding Shannon or R\'enyi entropies of the measurement outcomes. The entropic framework {
usually carries state-independent feature that} underlies cryptographic and quantum-memory applications~\cite{berta2010,Renner2005,Tomamichel2011PRL,Tomamichel2010IEEE,Hayashi2007}, while variance-based relations are geometrically intuitive and directly connected to spin squeezing and metrological precision~\cite{Giovannetti2006PRL,pezze2018,kitagawa1993,Hyllus2012}. Beyond their foundational role, both approaches serve as operational tools for entanglement detection and quantum-state characterization~\cite{hofmann2003,guhne2004,GuhneToth2009,Guhne2002PRA,Ren2021PRL}, and connect to modern notions of decoherence-free encoding and reference-frame-invariant information~\cite{lidar1998,bartlett2007}.

{
Despite their long history, genuinely state-independent bounds beyond small subsystems remain comparatively less developed within the variance-based, algebraic formulation. 
By contrast, the entropic paradigm has witnessed substantial multipartite generalizations and often displays state-independent information. 
Early correlation-based extensions were explored by Renes and Boileau~\cite{Renes2009}, who demonstrated how auxiliary correlations modify entropic uncertainty bounds. 
The introduction of quantum memory by Berta et al. subsequently provided a paradigmatic tripartite framework in which measurement uncertainty becomes state-independent once a correlated memory system is included.
Subsequent developments extended entropic uncertainty relations to more general multipartite and conditional scenarios~\cite{haddadi2021,Zhang2023MultipartiteEUR,Xu2025ECT}, and generalized multipartite entropic uncertainty relations have recently been experimentally investigated~\cite{Wang2024GEUR}
and also in other systems \cite{Wang2024SchwarzschildEUR}. 
These approaches typically rely on information-theoretic quantities such as conditional entropies or memory-assisted constructions.}

{
Beyond the two-particle level, however, no exact state-independent variance bounds were known prior to the present work, as the presence of multiple irreducible components and trivial sectors complicates the analysis~\cite{schwonnek2017}. This imbalance between the two paradigms motivates the present work.}

In this paper we develop a representation-theoretic framework for multipartite state-independent collective spin systems for variance-based uncertainty relations. For the standard $n$-qubit Hilbert space $(\mathbb{C}^2)^{\otimes n}$ carrying the $\mathfrak{su}_2$ representation generated by collective angular momentum operators $J_{x,y,z}$, we consider the total variance
\begin{equation}
\Delta^2(\mathfrak{su}_2)_v
= \sum_{a\in\{x,y,z\}}\!\big(\langle J_a^2\rangle_v-\langle J_a\rangle_v^2\big),
\end{equation}
and will show that it admits a minimal value. Our strategy is to converting the multipartite systems into direct-sum systems by the Clebsch--Gordan decomposition 
and then study the interaction between different sectors. The trivial representation $V^{(0)}$ permits complete suppression of collective variance, while nontrivial components enforce a strictly positive lower bound depending only on $j$.

Through explicit decomposition and extremal analysis, we obtain exact state-independent variance bounds up to $n=5$. For $n=1$, the fundamental bound $\Delta^2(\mathfrak{su}_2)\ge j$ reproduces the single-spin uncertainty. For $n=2$, $(\mathbb{C}^2)^{\otimes 2}\!\cong V^{(1)}\!\oplus V^{(0)}$, the trivial subspace allows $\Delta^2=0$ but the reduced space obeys $\Delta^2\!\ge\!1$. For $n=3$, the decomposition $(\mathbb{C}^2)^{\otimes 3}\!\cong V^{(3/2)}\!\oplus 2V^{(1/2)}$ yields the first fully positive universal bound $\Delta^2\!\ge\!4/11$. Extending to $n=4$, where $(\mathbb{C}^2)^{\otimes 4}\!\cong V^{(2)}\!\oplus 3V^{(1)}\!\oplus 2V^{(0)}$, we derive a reduced-space state-independent lower bound $\Delta^2\!\ge\!1/8$. 

{
To provide a compact roadmap of the results and to emphasize the emerging parity pattern, 
Table~\ref{tab:n-summary} summarizes the Clebsch--Gordan decompositions and the corresponding
(state-independent or reduced-space) variance bounds for $n=1$--$5$.}

\begin{table}[h]
\centering
\renewcommand{\arraystretch}{1.2}
\small
\begin{tabular}{c c c c}
\toprule
$n$-partite & Decomposition & Trivial sector? & State-independent bound \\
\midrule
1 & $V^{(1/2)}$ & No & $\Delta^2 \ge 1/2$ \\

2 & $V^{(1)}\oplus V^{(0)}$ & Yes & $\Delta^2 \ge 1$ on $V^{(1)}$ \\

3 & $V^{(3/2)}\oplus 2V^{(1/2)}$ & No & $\Delta^2 \ge 4/11$ \\

4 & $V^{(2)}\oplus 3V^{(1)}\oplus 2V^{(0)}$ & Yes & $\Delta^2 \ge 1/8$ on $V'$ \\

5 & $V^{(5/2)}\oplus4V^{(3/2)}\oplus5V^{(1/2)}$ & No & strictly positive (num. ind.)\\
\bottomrule
\end{tabular}
\caption{
Clebsch--Gordan decomposition and the corresponding variance-based uncertainty bounds for $n=1\sim 5$ qubit systems
under the collective $\mathfrak{su}_2$ action.
Odd subsystem numbers admit strictly positive universal bounds,
whereas even numbers contain trivial sectors permitting $\Delta^2=0$ on $V^{(0)}$ but retain positive reduced-space bounds.
}
\label{tab:n-summary}
\end{table}

{
As reflected in Table~\ref{tab:n-summary}, this sequence reveals an underlying parity-driven structural principle: even subsystem numbers
admit trivial representations, whereas odd subsystem numbers necessarily exclude
them, thereby enforcing strictly positive variance floors.}
These results bridge the previously separated algebraic and information-theoretic approaches to multipartite variance-based state-independent uncertainty relations, establishing the first explicit quadripartite variance-based state-independent uncertainty relation. 

Our representation-theoretic method further opens a way to study multipartite systems, providing a unified algebraic framework for state-independent uncertainty relations in multipartite qubit ensembles.
{
Thus, throughout this work, the Clebsch–Gordan decomposition 
serves not merely as a structural classification tool, 
but as the representation-theoretic mechanism 
through which state-independent collective uncertainty is determined.}

\section{State-independent Uncertainty Relations in One-partite Systems}

In this section, we study state-independent uncertainty relations (SIURs) for single quantum systems whose observables form a representation of the compact real Lie algebra \( \mathfrak{su}_2 \). This algebra is generated by three Hermitian operators \( J_x, J_y, J_z \) satisfying the standard angular momentum commutation relations:
\begin{equation}\label{e:su2}
[J_x, J_y] = i J_z, \quad [J_y, J_z] = i J_x, \quad [J_z, J_x] = i J_y.
\end{equation}

It is well-known that finite-dimensional 
irreducible unitary representations of $\mathfrak{su}_2$ are indexed by 
spins $\in \{ 0, \frac{1}{2}, 1, \frac{3}{2}, \dots \}$.
The spin $j$ irreducible representation $V^{(j)}$ has dimension \( \dim V^{(j)} = 2j + 1 \). In such representations, the quadratic Casimir operator ({\it which is $1/2$ of the usual Casimir operator for $\mathfrak{sl}_2$})
\begin{equation}\label{e:casimir}
C := J_x^2 + J_y^2 + J_z^2
\end{equation}
acts as scalar multiplication by \( j(j+1) \): for all $|v\rangle \in V^{(j)}$
\begin{equation}
C |v\rangle = j(j+1) |v\rangle  .
\end{equation}

The total variance of the algebra in a normalized quantum state \( |v\rangle \in V^{(j)} \) is defined by
\begin{equation}
  \Delta^2(\mathfrak{su}_2) := \Delta^2(J_x) + \Delta^2(J_y) + \Delta^2(J_z),  
\end{equation}
where each variance is given by
\begin{equation}
 \Delta^2(J_k) = \langle J_k^2 \rangle - \langle J_k \rangle^2, \quad k = x, y, z.   
\end{equation}
Using the Casimir identity, this can be written compactly as
\begin{equation}
    \Delta^2(\mathfrak{su}_2) = \langle C \rangle - \|\langle \vec{J} \rangle\|^2 = j(j+1) - \left( \langle J_x \rangle^2 + \langle J_y \rangle^2 + \langle J_z \rangle^2 \right).
\end{equation}

The total variance is minimized when the state maximizes the expectation of \(\vec{J}\) in some direction. In particular, for the highest weight state \( |j\rangle \), where \( J_z |j\rangle = j |j\rangle \) and \( \langle J_x \rangle = \langle J_y \rangle = 0 \), we obtain:
\begin{equation}
   \Delta^2(\mathfrak{su}_2) = j(j+1) - j^2 =j. 
\end{equation}
This is the minimal value the total variance can attain in \( V^{(j)} \).

We summarize this in the following result \cite{schwonnek2017}:

\begin{theorem}[State-independent Variance Bound]\label{thm:onepartite-su2}
Let \( |v\rangle \in V^{(j)} \) be a normalized state in the spin $j$ irreducible representation of \( \mathfrak{su}_2 \). Then
\begin{equation}\label{e:varbound}
\Delta^2(\mathfrak{su}_2) \geq j,
\end{equation}
with equality if and only if \( |v\rangle \) is an extremal weight vector, i.e., an eigenvector of \( J_z \) with eigenvalue \( \pm j \).
\end{theorem}

Such variance-based state-independent uncertainty relations for \( \mathfrak{su}_2 \) have been rigorously analyzed in \cite{deGuise2018}, where the minimal total variance is shown to be attained by coherent states, i.e., those aligned with extremal weights in the representation. 

In particular, the total variance \( \Delta^2(\mathfrak{su}_2) \) vanishes on the trivial representation \( V^{(0)} = \mathbb{C} \), since the action of \( \mathfrak{su}_2 \) is trivial. 
In general, if the space $V$ contains a trivial subrepresentation $V_0$ of $\mathfrak{su}_2$ (summation of $V^{(0}$), then 
the total variance $\Delta^2(\mathfrak{su}_2)$
is zero at any vector inside the trivial subspace $V_0$, or there are no uncertainty when the whole Lie algebra $\mathfrak{su}_2$ is measured with respect to any vector of the subspace $V_0$. 
It is thus natural to disregard the trivial subspace when we measure the total uncertainty. 

We therefore define the total variance of the Lie algebra $\mathfrak{su}_2$ over the space $V$ as
the total variance over the quotient space $\overline{V}=V/V_0$, where $V_0$ is the maximum subspace annihilated by $\mathfrak{su}_2$:
\begin{equation}\label{e:redvar1}
\Delta^2(\mathfrak{su}_2)_v := \Delta^2(\mathfrak{su}_2)_{v+V_0}
\end{equation}
over the quotient space $\overline{V}$.
Equivalently, this is the total variance computed within the subrepresentation space,
\begin{equation}\label{e:redvar2}
V' := V\ominus V_0,
\end{equation}
which is equivalent to the quotient space \(\overline{V}\) of $V$ mod out all trivial \( \mathfrak{su}_2 \)-submodules. We will call the quotient space $\overline{V}$ or equivalently the subspace $V'=V\ominus V_0$ the {\it reduced space}. Thus the total variance $\Delta^2(\mathfrak{su}_2)$ on $V$ is effectively the {\it reduced total variance} on $\overline{V}$. 

An equivalent formulation of the SIUR can be expressed in terms of first moments. Since the Casimir value is fixed and all generators are Hermitian, we have:
\begin{equation}\label{e:expbound}
\langle J_x \rangle^2 + \langle J_y \rangle^2 + \langle J_z \rangle^2 \leq j^2.
\end{equation}
This inequality is saturated precisely when the state is aligned with a maximal coherent state (highest or lowest weight). The bound corresponds to the squared length of the highest weight \( \lambda = j\alpha \), where \( (\alpha, \alpha) = 2 \), yielding \( (\lambda | \lambda) = 2 j^2 \).

These relations provide a clear state-independent constraint on both variances and expectation values for the single \( \mathfrak{su}_2 \) systems, i.e. one-partite system of irreducible module of $\mathfrak{su}_2 $. Extensions to multipartite systems will be developed in the following sections.

\section{Uncertainty Relations over Direct Sum Spaces}

In quantum systems governed by symmetry principles, composite state spaces may arise not only from tensor products but also from direct sum decompositions. These occur, for example, when a reducible representation is decomposed into irreducible subspaces according to symmetry breaking or spectral splitting. In this section, we focus on such direct sum representations and explore how the structure of uncertainty behaves across orthogonal components of the space.

\subsection{Variance Structure over Direct Sums}

We begin by establishing a general lower bound on the total variance for states residing in direct sum spaces. The following theorem provides a quantitative estimate based on the decomposition of the state vector.\\
\begin{lemma}\label{L:summand}
If for all unit vectors \( \|\bar{v}\| = 1 \), the variance satisfies
\begin{equation}
   \Delta^2(g)_{\bar{v}} := \sum_a\left(\langle \bar{v}| J_a^2 |\bar{v} \rangle - \langle \bar{v}| J_a |\bar{v} \rangle^2 \right) \ge j, 
\end{equation}
then for any $ v \in V$ with, $\|v\| \leq 1$, it holds that
\begin{equation}\label{eq:variance-bound}
    \sum_a\left(\langle v| J_a^2 |v \rangle - \langle v| J_a |v \rangle^2 \right) \ge j\|v\|^2.
\end{equation}
\end{lemma}

\begin{proof}
Let \( v \in V \) with \( \|v\| \le 1 \), and define the normalized vector \( \bar{v} := v / \|v\| \). Then:
\begin{align*}
\langle v | J_a | v \rangle &= \|v\|^2 \langle \bar{v} | J_a | \bar{v} \rangle, \\
\langle v | J_a^2 | v \rangle &= \|v\|^2 \langle \bar{v} | J_a^2 | \bar{v} \rangle.
\end{align*}
Hence the total variance satisfies:
    \begin{align*}
\Delta^2(g)_v 
&= \sum_a\left(\langle v| J_a^2|v \rangle - \langle v| J_a|v \rangle^2 \right) \\
&= \sum_a\left(\|v\|^2 \langle \bar{v}| J_a^2|\bar{v} \rangle - \|v\|^4 \langle \bar{v}| J_a|\bar{v} \rangle^2 \right) \\
&= \|v\|^2 \left( \sum_a \langle \bar{v}| J_a^2 |\bar{v} \rangle - \langle \bar{v}| J_a |\bar{v} \rangle^2 + (1 - \|v\|^2) \sum_a \langle \bar{v}| J_a |\bar{v} \rangle^2 \right) \\
&= \|v\|^2 \left( \Delta^2(g)_{\bar{v}} + (1 - \|v\|^2) \sum_a \langle \bar{v}| J_a |\bar{v} \rangle^2 \right) \\
&\ge \|v\|^2 \Delta^2(g)_{\bar{v}} \ge j \|v\|^2
\end{align*}
\end{proof}

\begin{remark}
The inequality $\Delta^2(g)_v \ge j \|v\|^2$ relies on the assumption $\|v\| \le 1$.
If $\|v\| > 1$, then the inequality~\eqref{eq:variance-bound} may fail as the proof argument shows.
\end{remark}
\begin{theorem}\label{thm:variance-extremal-bound}
Let \( V = \bigoplus_{k=1}^r V^{(j_k)} \) be a direct sum of finite-dimensional \( \mathfrak{su}_2 \)-modules, where each \( V^{(j_k)} \) is the irreducible spin-\( j_k \) representation. Let \( v = \sum_{k=1}^r v_k \in V \) be a normalized vector with \( \sum_{k=1}^r \|v_k\|^2 = 1 \), and \( v_k \in V^{(j_k)} \) denoting the orthogonal component of \( v \) in the \( k \)-th summand. 
Then the total variance satisfies the following universal lower bound:
\begin{equation}\label{eq:universal-variance-bound}
    \Delta^2(g)_v \ge \sum_{k=1}^r j_k \|v_k\|^2 - 2 \sum_{1 \le k < k' \le r} j_k j_{k'} \|v_k\|^2 \|v_{k'}\|^2.
\end{equation}
and the minimum of the right-hand side under the constraints \( \sum_{k=1}^r \|v_k\|^2 = 1 \) and \( \|v_k\|^2 > 0\) is given by
\begin{equation}\label{eq:f-min}
    f_{\min} = \frac{r}{4r - 4} + \left( \frac{1}{2} + \frac{ \sum_{i=1}^r \frac{1}{j_i} }{4 - 4r} \right) \cdot \frac{ -2(r-1) + \sum_{i=1}^r \frac{1}{j_i} }{ \left( \sum_{i=1}^r \frac{1}{j_i} \right)^2 - (r - 1) \sum_{i=1}^r \frac{1}{j_i^2} }.
\end{equation}

\end{theorem}
\begin{proof}
The total variance of \( v \in V \) is defined as
\[
\Delta^2(g)_v = \sum_{a=1}^3 \left( \langle J_a^2 \rangle_v - \langle J_a \rangle_v^2 \right),
\]
where \( \{J_1, J_2, J_3\} \) is a basis of \( \mathfrak{su}_2 \). Since the \( J_a \) preserve the decomposition \( V = \bigoplus_k V^{(j_k)} \), for any normalized vector $v=\sum_{k=1}^k v_k$, $v_k\in V^{(j_k)}$ one has that:
\begin{equation}
\langle J_a^2 \rangle_v = \sum_{k=1}^r \langle J_a^2 \rangle_{v_k}, \quad
\langle J_a \rangle_v = \sum_{k=1}^r \langle J_a \rangle_{v_k}.
\end{equation}
Expanding the variance:
\begin{align*}
\Delta^2(g)_v
&= \sum_{a=1}^3 \left( \sum_k \langle J_a^2 \rangle_{v_k} - \left( \sum_k \langle J_a \rangle_{v_k} \right)^2 \right) \\
&= \sum_{k=1}^r \sum_{a=1}^3 \left( \langle J_a^2 \rangle_{v_k} - \langle J_a \rangle_{v_k}^2 \right)
 - 2 \sum_{1 \le k < k' \le r} \sum_{a=1}^3 \langle J_a \rangle_{v_k} \langle J_a \rangle_{v_{k'}} \\
&= \sum_{k=1}^r \Delta^2(g)_{v_k} - 2 \sum_{1 \le k < k' \le r} \sum_{a=1}^3 \langle J_a \rangle_{v_k} \langle J_a \rangle_{v_{k'}}.
\end{align*}
Note that $\|v_k\|\leq 1$, it follows from Theorem \ref{thm:onepartite-su2} and Lemma \ref{L:summand} that
\begin{equation}
    \sum_{k=1}^r \Delta^2(g)_{v_k} \ge \sum_{k=1}^r j_k \|v_k\|^2.
\end{equation}
To bound the cross terms, apply the Cauchy–Schwarz inequality:
\begin{equation}
   \sum_{a=1}^3 \langle J_a \rangle_{v_k} \langle J_a \rangle_{v_{k'}} 
\le \left( \sum_{a=1}^3 \langle J_a \rangle_{v_k}^2 \right)^{1/2} \left( \sum_{a=1}^3 \langle J_a \rangle_{v_{k'}}^2 \right)^{1/2}. 
\end{equation}
Moreover, for each \( k \), we use the inequality:
\begin{equation}
  \sum_{a=1}^3 \langle J_a \rangle_{v_k}^2 \le j_k^2 \|v_k\|^4,  
\end{equation}
which follows from the general bound \( \|\langle \vec{J} \rangle_{v_k}\| \le j_k \|v_k\|^2 \) (cf. \eqref{e:expbound}). Hence:
\begin{equation}
  \sum_{a=1}^3 \langle J_a \rangle_{v_k} \langle J_a \rangle_{v_{k'}}
\le j_k j_{k'} \|v_k\|^2 \|v_{k'}\|^2.  
\end{equation}
Putting all together:
\begin{equation}
  \Delta^2(g)_v \ge \sum_{k=1}^r j_k \|v_k\|^2 - 2 \sum_{1 \le k < k' \le r} j_k j_{k'} \|v_k\|^2 \|v_{k'}\|^2.  
\end{equation}
{
The second term
\begin{equation}
   2 \sum_{1 \le k < k' \le r} j_k j_{k'} \|v_k\|^2 \|v_{k'}\|^2 
\end{equation}
arises from the expansion of the squared collective expectation value 
$\|\langle \vec{J} \rangle_v\|^2 = \|\sum_k \langle \vec{J} \rangle_{v_k}\|^2$. 
It represents the interference between different irreducible components in the direct-sum decomposition. 
While each sector individually enforces a positive variance floor proportional to $j_k$, 
the cross terms may reduce the total variance when expectation values from different sectors align coherently. 
This mechanism is responsible for the nontrivial structure of the multipartite lower bounds.}

Letting \( x_k = \|v_k\|^2 >0\), which satisfies the normalization constraint \( \sum_{k=1}^r x_k = 1 \), so we need to minimize the auxiliary function:
\begin{equation}
   f(x) = \sum_{k=1}^r j_k x_k - 2 \sum_{1 \le i < j \le r} j_i j_j x_i x_j , 
\end{equation}
under the constraint $\sum_k x_k=1$,
where \( j_k\) are the spins. 

By Lemma \ref{l:min} in the Appendix, the minimal value of the objective function is given by
\begin{equation}
  f_{\min} = \frac{r}{4r - 4} + \left( \frac{1}{2} + \frac{ \sum_{i=1}^r \frac{1}{j_i} }{4 - 4r} \right) \cdot \frac{ -2(r+1) + \sum_{i=1}^r \frac{1}{j_i} }{ \left( \sum_{i=1}^r \frac{1}{j_i} \right)^2 - (r - 1) \sum_{i=1}^r \frac{1}{j_i^2} },  
\end{equation}
which proves the result.
\end{proof}
\begin{corollary}[Global Minimum on the Boundary of the Simplex]\label{cor:boundary-global-min}
Let \( V = \bigoplus_{k=1}^r V^{(j_k)} \), \( v = \sum_{k=1}^r v_k \in V \), and \( x_k = \|v_k\|^2 \) be as in Theorem~\ref{thm:variance-extremal-bound}, with \( x_k \ge 0 \) and \( \sum_{k=1}^r x_k = 1 \). Then the global minimum of the lower bound function
\[
f(x) = \sum_{k=1}^r j_k x_k - 2 \sum_{1 \le i < j \le r} j_i j_j x_i x_j ,
\]
over the closed simplex
\[
\mathcal{S}_r := \left\{ x \in \mathbb{R}^r \,\middle|\, x_k \ge 0,\ \sum_{k=1}^r x_k = 1 \right\}
\]
is achieved by minimizing over all nonempty support subsets of \( \mathcal{S}_r \), and is given by
\begin{equation}\label{eq:f-global-min}
f_{\min}^{\mathrm{global}} = \min_{\emptyset \ne S \subseteq \{1, \dots, r\}} f_{\min}^{(S)},
\end{equation}
where \( f_{\min}^{(S)} \) denotes the value given by
\eqref{eq:f-min} with the spins $j_k, k\in S$, which is also the local minimum value under the constraint that \( x_k = 0 \) for all \( k \notin S \) and \( \sum_{k \in S} x_k = 1 \).

More precisely, using only \( \{j_i\}_{i \in S} \), we have
\begin{equation}\label{eq:f-min-S}
    f_{\min}^{(S)} = \frac{|S|}{4|S| - 4} + \left( \frac{1}{2} + \frac{ \sum_{i \in S} \frac{1}{j_i} }{4 - 4|S|} \right) \cdot \frac{ -2(|S|+1) + \sum_{i \in S} \frac{1}{j_i} }{ \left( \sum_{i \in S} \frac{1}{j_i} \right)^2 - (|S| - 1) \sum_{i \in S} \frac{1}{j_i^2} },
\end{equation}
where $|S|$ is the cardinality of $S$.
\end{corollary}
\begin{proof}
Theorem~\ref{thm:variance-extremal-bound} computes the minimum of \( f(x) \) under the assumption that all components \( x_k > 0 \), corresponding to interior points of the simplex. Observe that the boundary condition $x_i=0$ corresponds exactly letting the spin $j_i=0$. So the argument of Theorem~\ref{thm:variance-extremal-bound} also extends to the boundary cases where some \( x_k = 0 \), enabling a global minimum search over the full simplex \( \mathcal{S}_r \). For any such subset \( S \), the corresponding spin labels \( j_k \) with \( k \notin S \) have no influence on \( f(x) \) and can be excluded from the evaluation of the minimum.
\end{proof}

\begin{example}[Direct Sum $V^{(\tfrac{3}{2})} \oplus V^{(\tfrac{1}{2})}$]
We consider the direct sum of two irreducible \( \mathfrak{su}_2 \)-modules: \( V = V^{\left(\tfrac{3}{2}\right) }\oplus V^{\left(\tfrac{1}{2}\right)} \). For any normalized state \( v = v_1 + v_2 \in V \), where \( v_1 \in V^{\left(\tfrac{3}{2}\right) }\), \( v_2 \in V^{\left(\tfrac{1}{2}\right)} \), and \( \|v_1\|^2 + \|v_2\|^2 = 1 \), the total angular momentum variance satisfies
\[
\Delta^2(\mathfrak{su}_2)_v \ge f(x) = \frac{3}{2}x + \frac{1}{2}(1-x)  - 2\cdot\frac{3}{2}\cdot\frac{1}{2}x(1-x) \quad\text{where } x = \|v_1\|^2.
\]
Applying Theorem~\ref{thm:variance-extremal-bound}, the minimum of \( f(x) \) over \( x \in [0,1] \) is given by
\[
f_{\min} = \frac{1}{2} + \left( \frac{1}{2} \cdot \frac{3}{2} \cdot \frac{1}{2} - \frac{1}{4}\cdot (\frac{3}{2} + \frac{1}{2}) \right) \cdot \left( -1 + \frac{1}{2} \left( \frac{2}{3} + 2 \right) \right) = \frac{11}{24}
\]
Hence, we obtain the sharp bound:
\[
\Delta^2(\mathfrak{su}_2)_v \ge  \frac{11}{24} \approx 0.4583 \quad \text{for all normalized } v \in V.
\]
\end{example}

This example demonstrates the concrete strength of our lower bound in reducible representations of \( \mathfrak{su}_2 \), where the interference between components significantly affects the minimal uncertainty.
     
\section{Bipartite Uncertainty Relations}
{
The relevance of the Clebsch–Gordan decomposition 
stems from the fact that the collective $\mathfrak{su}_2$ generators 
act block-diagonally on irreducible subrepresentations. 
Since the total variance is defined entirely in terms of these generators, 
its minimal achievable value is determined by how a quantum state 
is distributed across irreducible spin sectors. 
In particular, the presence of a trivial representation $V^{(0)}$ 
permits complete suppression of collective fluctuations, 
whereas nontrivial irreducible components enforce strictly positive 
representation-level variance floors.}


Throughout this section, we let \( \mathfrak{g} = \mathfrak{su}_2 \), the compact real form of \( \mathfrak{sl}_2(\mathbb{C}) \). All representations are finite-dimensional, unitary, and the total variance is computed with respect to an orthonormal basis \( \{J_x, J_y, J_z\} \) of \( \mathfrak{g} \) under the standard inner product.

\begin{theorem}[Relative Positivity of Bipartite Variance]\label{thm:bipartite}
Let \( V = V^{\left(\tfrac{1}{2}\right)} \otimes V^{\left(\tfrac{1}{2}\right)} \cong V^{\left(1\right)} \oplus V^{\left(0\right)} \) be the bipartite representation of \( \mathfrak{su}_2 \). For any normalized state \( v\) in the reduced space $\overline{V}\simeq  V^{(1)}$, the (reduced) total variance satisfies
\begin{equation}
  \Delta^2(\mathfrak{su}_2)_v \ge 1.  
\end{equation}
\end{theorem}

\begin{proof}
From the Clebsch--Gordan decomposition, we have:
\begin{equation}
  V^{\left(\tfrac{1}{2}\right)} \otimes V^{\left(\tfrac{1}{2}\right)} \cong V^{\left(1\right)} \oplus V^{\left(0\right)} .  
\end{equation}
Equivalently, in terms of Young diagrams,
\begin{equation}
  \ytableaushort{\ }
\mathbin{\otimes}
\ytableaushort{\ }
\cong 
\ytableaushort{\ \ }
\oplus
\ytableaushort{\ ,\ }
\cong
\ytableaushort{\ \ }
\oplus
\bullet  
\end{equation}
where 
\(
\ytableaushort{\ \ }
\)
corresponds to $V^{(1)}$ and $\bullet$
corresponds to the trivial representation $V^{(0)}$.

The variance vanishes exactly when \( v\in V^{(0)}\), and the reduced variance is bounded below by $1$.
The result immediately follows from Theorem
\ref{thm:onepartite-su2}.
\end{proof}

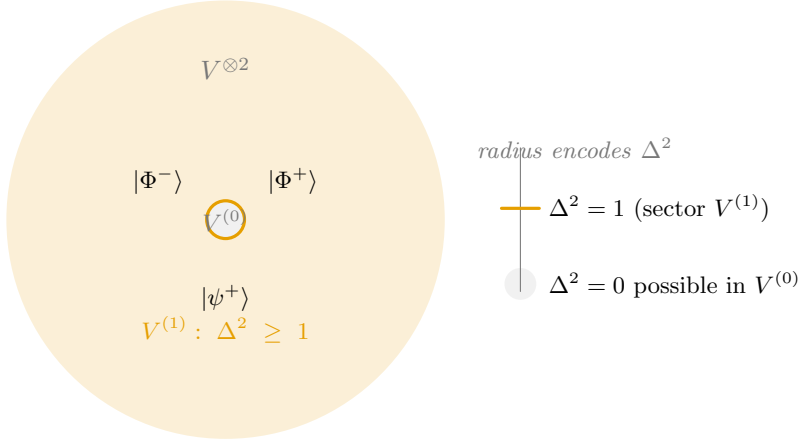
\begin{figure}[h]
\centering
\begin{tikzpicture}[>=Latex,line cap=round,line join=round,
  every node/.style={font=\small}]
  \definecolor{oiOrange}{RGB}{230,159,0}
  \definecolor{oiGreen}{RGB}{0,158,115}
  \definecolor{oiGrey}{RGB}{120,120,120}

  \coordinate (C) at (0,0);
  \def\R{2.9}   
  \def\r{0.25}  

  \fill[oiOrange,opacity=0.16] (C) circle (\R);

  \fill[white] (C) circle (\r);              
  \fill[oiGrey,opacity=0.10] (C) circle (\r);
  \draw[oiOrange,very thick] (C) circle (\r);
  \node[oiGrey] at (0,0) {$V^{(0)}$};


\node[oiGrey,font=\bfseries] at (0,2.0) {$V^{\otimes 2}$};

  \node[oiOrange] at (0,-\r-1.2) {$V^{(1)}:\ \Delta^2\ \ge\ 1$};

  \node at ({(\r+0.78)*cos(30)},  {(\r+0.78)*sin(30)})  {$|\Phi^+\rangle$};
  \node at ({(\r+0.78)*cos(150)}, {(\r+0.78)*sin(150)}) {$|\Phi^-\rangle$};
  \node at ({(\r+0.78)*cos(270)}, {(\r+0.78)*sin(270)}) {$|\psi^+\rangle$};

  \draw[oiGrey] (3.9,-0.95) -- (3.9, 0.95);
  \fill[oiGrey,opacity=0.10] (3.9,-0.85) circle (6pt);
  \node[anchor=west] at (4.15,-0.85) {$\Delta^2=0$ possible in $V^{(0)}$};
  \draw[oiOrange,very thick] (3.65,0.15) -- (4.15,0.15);
  \node[anchor=west] at (4.15,0.15) {$\Delta^2=1$ (sector $V^{(1)}$)};
  \node[oiGrey,align=left,anchor=west] at (3.2,0.95)
    {\textit{radius encodes} $\Delta^2$};

\end{tikzpicture}

\caption{
Bipartite structure and the corresponding state-independent variance bound. 
The radial coordinate encodes the total collective-spin variance $\Delta^2$ for $(\tfrac{1}{2})\!\otimes\!(\tfrac{1}{2}) = 1 \oplus 0$. 
The central subspace $V^{(0)}$ (grey) is fixed at $\Delta^2=0$, while all states in $V^{(1)}$ (orange region) satisfy $\Delta^2\!\ge\!1$. 
Representative Bell states $|\Phi^\pm\rangle$ and $|\psi^+\rangle$ are indicated on the ring. 
Radius encodes the magnitude of total variance.
}

\end{figure}
The bipartite Hilbert space $(\mathbb{C}^2)^{\otimes 2}\cong V^{(\tfrac12)}\!\otimes\! V^{(\tfrac12)}$ carries the diagonal $\mathfrak{su}_2$ action and decomposes as $V^{(1)}\!\oplus\!V^{(0)}$. 
The subspace $V^{(0)}=\operatorname{span}\{\tfrac{1}{\sqrt{2}}(|01\rangle-|10\rangle)\}$ is invariant under all global $\mathfrak{su}_2$ rotations and thus exhibits vanishing total variance for every collective generator. 
In contrast, any state with nonzero support on the subspace $V^{(1)}$ obeys a strictly positive, state-independent lower bound (Theorem~\ref{thm:bipartite}). 
Notably, $V^{(1)}$ contains both separable and maximally entangled Bell states, implying that the ability to suppress global $\mathfrak{su}_2$ uncertainty
is governed by the representation sector rather than
by entanglement alone.

\medskip
This structural distinction has direct physical implications. 
The sector of trivial representations constitutes a natural decoherence-free subspace under collective $\mathrm{SU}(2)$ noise, forming the basis for robust architectures in quantum computation~\cite{lidar1998}. 
Moreover, its rotational invariance enables reference-frame–independent communication and quantum key distribution~\cite{bartlett2007}. 
For broader background on uncertainty relations in quantum information, see \cite{wehner2010}.

\section{Tripartite Uncertainty Relations}
{
In contrast to the bipartite case, 
we now consider the tripartite representation 
$V=(\mathbb{C}^2)^{\otimes3}$ under the collective $\mathfrak{su}_2$ action. 
The Clebsch–Gordan decomposition reveals whether trivial sectors 
persist at this level and therefore whether collective variance 
can in principle vanish. 
The absence or presence of such sectors 
fully determines the possibility of a state-independent positive floor.}

\begin{theorem}[State-Independent Variance Bound in the Tripartite Qubit System]
\label{thm:tripartite-variance}
Let \( \mathfrak{g} = \mathfrak{su}_2 \), and consider the tripartite spin representation \( V = (\mathbb{C}^2)^{\otimes 3} \) equipped with the total angular momentum action. Then for any normalized state \( v \in V \), the total variance of \( \mathfrak{g} \) satisfies the universal, state-independent lower bound
\begin{equation}
\Delta^2(\mathfrak{g})_v \ge \frac{4}{11} \approx 0.3636.
\end{equation}
\end{theorem}
\begin{proof}
By successive applications of the Clebsch--Gordan decomposition, the representation space \( V = (\mathbb{C}^2)^{\otimes 3} \) decomposes into irreducible \( \mathfrak{su}_2 \)-modules as
\begin{equation}
 V \cong V^{\left(\tfrac{3}{2}\right) }\oplus 2V^{\left(\tfrac{1}{2}\right)}.   
\end{equation}
{
In terms of Young diagrams,
\begin{equation}
 \ytableaushort{\ }
\otimes
\ytableaushort{\ }
\otimes
\ytableaushort{\ }
\cong
\ytableaushort{\ \ \ }
\oplus
2\,\ytableaushort{\,},   
\end{equation}
where
\(
\ytableaushort{\ \ \ }
\)
denotes $V^{(3/2)}$ and 
\(
\ytableaushort{\,}
\)
means $V^{(1/2)}$.
}

Applying Theorem~\ref{thm:variance-extremal-bound} with \( j_1 = \tfrac{3}{2} \), \( j_2 = j_3 = \tfrac{1}{2} \), \( r = 3\), and using the formula~\eqref{eq:f-min} for the minimum, we obtain
\[
\Delta^2(\mathfrak{g})_v \ge \frac{4}{11}
\]
for any normalized state $v\in V$.
\end{proof}

\begin{figure}[!ht]
\centering
\begin{tikzpicture}[>=Latex,line cap=round,line join=round,
  every node/.style={font=\small}]
  \definecolor{oiBlue}{RGB}{0,114,178}
  \definecolor{oiOrange}{RGB}{230,159,0}
  \definecolor{oiGrey}{RGB}{120,120,120}
  \definecolor{oiGreen}{RGB}{0,158,115}

  \coordinate (C) at (0,0);
  \def\R{2.9}    
  \def\rT{1.25}  

  \fill[oiOrange,opacity=0.16] (C) circle (\R);
  \node[oiOrange,anchor=west] at ({\rT*cos(8)}, {\rT*sin(8)}) {$\Delta^2=\tfrac{4}{11}$};

\node[oiGrey,font=\bfseries] at (0,2.0) {$V^{\otimes 3}$};

  \node[oiBlue] at ({(\rT+0.7)*cos(45)}, {(\rT+0.7)*sin(45)}) {$|\mathrm{GHZ}\rangle$};
  \node[oiGreen] at ({(\rT+0.7)*cos(225)}, {(\rT+0.7)*sin(225)}) {$|\mathrm{W}\rangle$};

\draw[oiGrey] (3.8,-0.9) -- (3.8,0.9);


\fill[oiOrange,opacity=0.10] (3.75,0.2) circle (6pt);
\node[anchor=west] at (3.95,0.2) {$\Delta^2=\tfrac{4}{11}$};

\node[oiGrey,align=left,anchor=west] at (3.2,0.9)
  {\textit{radius encodes} $\Delta^2$};

\end{tikzpicture}

\caption{
Visualization of the tripartite system $(\tfrac12)^{\!\otimes3}\!\cong\!V^{(3/2)}\!\oplus\!2V^{(1/2)}$
and the corresponding state-independent lower bound on the total variance $\Delta^2$.
Since no sector $V^{(0)}$ appears at $N{=}3$, the total variance cannot vanish.
All states obey the uniform bound $\Delta^2\ge\tfrac{4}{11}$ (orange disk),
with radius encoding the variance magnitude.
Representative entangled states $|\mathrm{GHZ}\rangle$ (blue) and $|\mathrm{W}\rangle$ (green)
are indicated on the ring to illustrate distinct coherence structures within the same bound.
}
\end{figure}
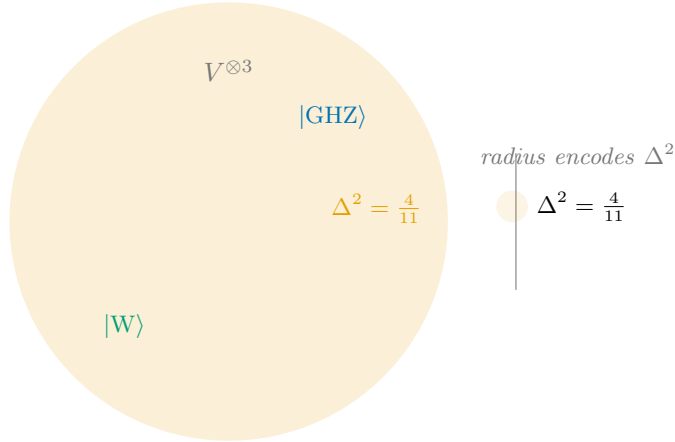
The tripartite Hilbert space $(\mathbb{C}^2)^{\otimes 3}\cong V^{(\tfrac{3}{2})}\!\oplus\!2V^{(\tfrac{1}{2})}$ contains no trivial irreducible subrepresentation. 
Consequently, no normalized state can suppress the global $\mathfrak{su}_2$ variance to zero. 
Using the Clebsch--Gordan decomposition together with Theorem~\ref{thm:variance-extremal-bound}, we obtain a universal, state-independent lower bound valid for all tripartite states:
\[
\Delta^2(\mathfrak{su}_2)_v \ge \tfrac{4}{11}.
\]
This contrasts with the bipartite case, where projection onto the subspace $V^{(0)}$ completely eliminates collective variance. 
Representative tripartite states such as $|\mathrm{GHZ}\rangle=\tfrac{1}{\sqrt{2}}(|000\rangle+|111\rangle)$ and $|\mathrm{W}\rangle=\tfrac{1}{\sqrt{3}}(|100\rangle+|010\rangle+|001\rangle)$ lie entirely outside any trivial sector and therefore obey a strictly positive uncertainty floor. 
In practice, the absence of a sector of trivial subrepresentations precludes subspace-based decoherence-free coding under fully collective $\mathrm{SU}(2)$ noise and instead motivates noiseless subsystem encodings and collective-spin diagnostics~\cite{dammeier2015,deGuise2018,lidar1998}.

\section{Quadripartite Uncertainty Relations}
{
We now extend this representation-theoretic analysis to the quadripartite system. 
The Clebsch–Gordan decomposition becomes increasingly structured 
as multiplicities grow, and the interplay between trivial and nontrivial 
irreducible sectors governs the attainable collective-spin variance. 
The decomposition therefore provides a direct algebraic criterion 
for identifying whether zero-variance states exist and whether 
a reduced-space positive bound must emerge.}


The Clebsch--Gordan coefficients imply that
\begin{align}
    \left( \mathbb{C}^2 \right)^{\otimes 4}
    &\cong V^{(2)} \oplus 3V^{(1)} \oplus 2V^{(0)}.
\end{align}
The decomposition is obtained by the following Young diagram computation:
{
\begin{align}\label{e:tensordec}
\ytableaushort{\ \ \ \ }
&\cong
(\ytableaushort{\ \ \ }
\oplus
2\,\ytableaushort{\  ,})
\otimes
\,\ytableaushort{\ }   \\
&\cong
(\ytableaushort{\ \ \ \ }
\oplus
\ytableaushort{\ \ ,})
\oplus
2(\ytableaushort{\ \  }
\oplus
\bullet)\\
&\cong 
\ytableaushort{\ \ \ \ }
\oplus
3\,\ytableaushort{\ \ ,}
\oplus
2\bullet.
\end{align}
Here $\ytableaushort{\ \ \ \ }, \ytableaushort{\ \ }$, and $\bullet$  correspond respectively  to $V^{(2)}$, $V^{(1)}$, and 
the trivial representation.
}

The presence of two copies of the trivial representation \(V^{(0)}\) implies that complete variance suppression is again possible---as in the bipartite case---but only for states fully supported outside this null sector. For all other configurations, nonzero uncertainty is enforced by irreducibility and interference.

\begin{theorem}[Variance Bound with Interference in the Quadripartite System]
Let \( \mathfrak{g} = \mathfrak{su}_2 \), and consider the quadripartite space
\begin{equation}
  V = (\mathbb{C}^2)^{\otimes 4} \cong V^{(2)} \oplus 3V^{(1)} \oplus 2V^{(0)}.  
\end{equation}
Then \( \overline{V}\simeq V^{(2)} \oplus 3V^{(1)} \), and the reduced total variance
has a state-independent lower bound
\begin{equation}
\Delta^2(\mathfrak{g})_{\overline{V}}\ge \frac{1}{8}.
\end{equation}
\end{theorem}
\begin{proof}
We consider the orthogonal decomposition
\[
V = (\mathbb{C}^2)^{\otimes 4} \cong V' \oplus V_0, \quad \text{where } V'= V^{(2)} \oplus 3V^{(1)},\quad V_0 = 2V^{(0)}.
\]

To get the reduced total variance, we apply Theorem~\ref{thm:variance-extremal-bound} to the space
\begin{equation}
 V' = V^{(2)} \oplus V^{(1)} \oplus V^{(1)} \oplus V^{(1)}.   
\end{equation}
using
\[
r = 4,\quad (j_1, j_2, j_3, j_4) = (2, 1, 1, 1),
\]
we obtain that the minimal value over normalized vectors in \( V' \) as
\[
f_{\min}^{(V')} = \frac{1}{8}.
\]
\end{proof}

In other words, for any normalized state 
\( v \in V' \oplus V_0\) decomposes as
\[
v = v_1 + v_2, \quad \text{with } \|v_1\|^2 + \|v_2\|^2 = 1,
\]
where $v_1\in V', v_2\in V^{(0)}$.
Then the total variance satisfies the inequality:
\begin{equation}
  \Delta^2(\mathfrak{g})_v \ge \frac{1}{8}(1 - \|v_2\|^2).  
\end{equation}

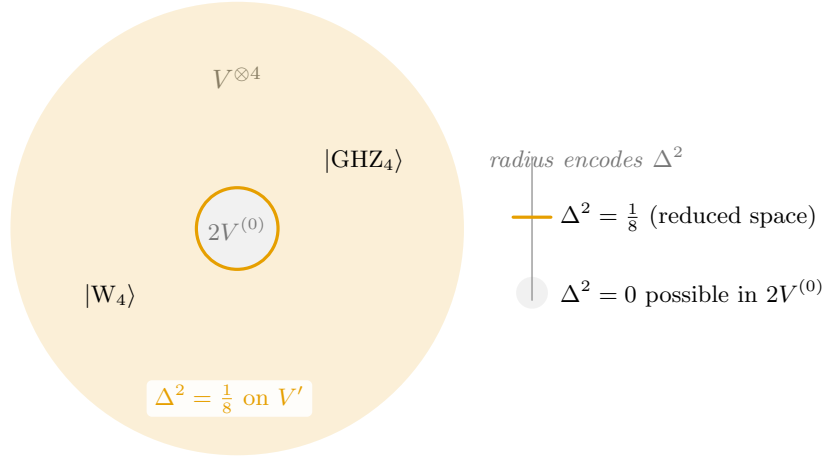
\begin{figure}[h]
\centering
\begin{tikzpicture}[>=Latex,line cap=round,line join=round,
  every node/.style={font=\small}]
  \definecolor{oiOrange}{RGB}{230,159,0}
  \definecolor{oiGrey}{RGB}{120,120,120}
  \definecolor{oiGreen}{RGB}{0,158,115}

  \coordinate (C) at (0,0);
  \def\R{3.0}       
  \def\rT{1.35}     
  \def\rGrey{0.4*\rT} 

\node[oiGrey,font=\bfseries] at (0,2.0) {$V^{\otimes 4}$};
  \fill[oiOrange,opacity=0.16] (C) circle (\R);

  \fill[white] (C) circle (\rGrey);                   
  \fill[oiGrey,opacity=0.10] (C) circle (\rGrey);     
  \node[oiGrey] at (0,0) {$2V^{(0)}$};

  \draw[oiOrange,very thick] (C) circle (\rGrey);
  \node[oiOrange,anchor=north east,fill=white,fill opacity=.9,text opacity=1,
        rounded corners=2pt,inner xsep=3pt,inner ysep=1.5pt]
        at ($ (C) + (1,-2) $)
        {$\Delta^2=\tfrac{1}{8}$ on $V'$};

  \node at ({(\rT+0.55)*cos(28)},  {(\rT+0.55)*sin(28)})  {$|\mathrm{GHZ}_4\rangle$};
  \node at ({(\rT+0.55)*cos(208)}, {(\rT+0.55)*sin(208)}) {$|\mathrm{W}_4\rangle$};

  \draw[oiGrey] (3.9,-0.95) -- (3.9, 0.95);
  \fill[oiGrey,opacity=0.10] (3.9,-0.85) circle (6pt);
  \node[anchor=west] at (4.15,-0.85) {$\Delta^2=0$ possible in $2V^{(0)}$};
  \draw[oiOrange,very thick] (3.65,0.15) -- (4.15,0.15);
  \node[anchor=west] at (4.15,0.15) {$\Delta^2=\tfrac{1}{8}$ (reduced space)};
  \node[oiGrey,align=left,anchor=west] at (3.2,0.95)
    {\textit{radius encodes} $\Delta^2$};

\end{tikzpicture}

\caption{Quadripartite $\mathfrak{su}_2$ structure and the corresponding state-independent variance bound.
The radial coordinate represents the total variance $\Delta^2$ of the collective spin for $(\mathbb{C}^2)^{\otimes4}$.
The shaded central disk (scaled to 85\% of the original radius) denotes the two-dimensional subspace $2V^{(0)}$.
The orange circle marks the state-independent lower bound $\Delta^2=\tfrac{1}{8}$ on the reduced nontrivial space $V' = V^{(2)} \oplus 3V^{(1)}$.
Representative states $|\mathrm{GHZ}_4\rangle$ and $|\mathrm{W}_4\rangle$ are indicated along the ring.}

\end{figure}
The quadripartite space $(\mathbb{C}^2)^{\otimes 4}\cong V^{(2)}\oplus 3V^{(1)}\oplus 2V^{(0)}$ contains both trivial and nontrivial irreps. Hence complete suppression of the global $\mathfrak{su}_2$ variance is possible only for states fully supported on the two dimensional trivial sector $2V^{(0)}$. When we restrict to the reduced space $V'=V^{(2)}\oplus 3V^{(1)}$, our analysis yields a uniform, state independent lower bound on the total variance, namely $\Delta^2(\mathfrak{su}_2)_v\ge \tfrac{1}{8}$ for all normalized $v\in V'$. This situates the quadripartite case between the bipartite scenario, where a trivial subspace can drive the variance to zero, and the tripartite scenario, where a positive floor holds for every state. 

Operationally, the presence of $2V^{(0)}$ enables decoherence free encoding for fully collective $\mathrm{SU}(2)$ noise, while the reduced space bound provides a baseline for collective spin based witnesses and metrological benchmarks in finite dimensional $\mathrm{SU}(2)$ systems \cite{lidar1998,dammeier2015,deGuise2018}. {
For comparison with entropic formulations,
state-independent bounds based on quantum memory
were first established in the tripartite scenario
\cite{berta2010,coles2017}
and have since been extended to more general
multipartite settings
\cite{Zhang2023MultipartiteEUR,Wang2024GEUR,Xu2025ECT}.
In contrast, variance-based bounds exhibit their most
transparent algebraic structure in the bipartite case,
where the presence of trivial representations
directly governs the possibility of variance suppression.
Our quadripartite result reveals an intermediate structure:
while trivial sectors reappear,
a strictly positive reduced-space bound persists.}

\section{Quintipartite Uncertainty Relations}
{
At the quintipartite level, the representation structure becomes highly reducible, 
with substantial multiplicities among lower-spin sectors. 
Since collective generators act block-diagonally on each irreducible component, 
the minimal total variance is determined by interference 
between these sectors. 
The Clebsch–Gordan decomposition thus furnishes the algebraic framework 
within which we analyze possible variance suppression.}

We now extend the representation-theoretic analysis to the quintipartite system
$
V=(\mathbb{C}^2)^{\otimes5}
$.
Its Clebsch--Gordan decomposition leads to
\begin{equation}
  (\mathbb{C}^2)^{\otimes5}\cong
V^{(5/2)}\oplus4V^{(3/2)}\oplus5V^{(1/2)},  
\end{equation}
comprising ten irreducible sectors ($r=10$) with total spins $j\in\{\tfrac{5}{2},\tfrac{3}{2},\tfrac{1}{2}\}$.

The decomposition is seen similarly as in \eqref{e:tensordec} using Young diagrams.In terms of Young diagrams,
{
\begin{equation}
  \ytableaushort{\ \ \ \ \ }
\cong
\ytableaushort{\ \ \ \ \ }
\oplus
4\,\ytableaushort{\ \ \ ,}
\oplus
5\,\ytableaushort{\ }, 
\end{equation}
}
where
\(
\ytableaushort{\ \ \ \ \ }
\)
corresponds to $V^{(5/2)}$,
\(
\ytableaushort{\ \ \ }
\)
corresponds to $V^{(3/2)}$,
and 
\(
\ytableaushort{\,}
\)
corresponds to $V^{(1/2)}$.
\medskip

Applying Theorem~\ref{thm:variance-extremal-bound} and Corollary~\ref{cor:boundary-global-min} to a general normalized state
$v=\sum_{k=1}^{10}v_k$, with $v_k\in V^{(j_k)}$ and
$x_k=\|v_k\|^2$ ($\sum_kx_k=1$, $x_k\ge0$), yields a quadratic lower bound on
$\Delta^2(\mathfrak{su}_2)_v$. Two regimes emerge:

\begin{itemize}
\item[(i)] \textit{Interior solution:}
All $x_k>0$ (no vanishing weights) gives a strictly positive minimum
\[
f_{\mathrm{int}}^{\min}\approx0.3031.
\]

\item[(ii)] \textit{Boundary solution:}
Allowing some $x_k=0$ yields a global minimum
$f_{\mathrm{bdry}}^{\min}\approx-0.6075$, 
attained for
\[
x_1\simeq0.269\ (j=\tfrac{5}{2}),\quad
x_{2,3,4,5}\simeq0.183\ (j=\tfrac{3}{2}),\quad
x_{6},\ldots,x_{10}=0\ (j=\tfrac{1}{2}).
\]
This negative value indicates that the quadratic relaxation overestimates cross-term subtraction when several degenerate subspaces overlap, thus ceasing to be a tight physical bound.
\end{itemize}
{
We emphasize that the total variance is nonnegative by definition, 
since it is a sum of variances of Hermitian operators. 
The appearance of a negative value therefore does not signal a physical violation, 
but rather reflects the fact that the quadratic relaxation employed in 
Theorem~\ref{thm:variance-extremal-bound} is not tight when large 
multiplicities and overlapping irreducible components are present. 
In particular, interference terms between degenerate low-spin sectors 
are over-subtracted in the relaxed optimization problem. 
This suggests that the true minimal collective-spin variance
in the full quintipartite space remains nonnegative,
and we conjecture that it may in fact be strictly positive.}

\medskip
To further isolate the physically meaningful regime, we examined partial subspaces obtained by excluding selected irreducible components and re-evaluating the minimal achievable variance.
The results show that positive bounds reappear once sufficiently many low-spin sectors are removed. Specifically, when the subspaces $V^{(5/2)}\oplus2V^{(3/2)}$ are excluded, the minimal variance becomes strictly positive ($f_{\min}=0.2857$).

The following result follows from a careful application of Corollary
\ref{cor:boundary-global-min}.
\begin{theorem}[Reduced-space positivity for the quintipartite system]
Let 
\begin{equation}
  V=(\mathbb{C}^2)^{\otimes5}\cong
V^{(5/2)}\oplus4V^{(3/2)}\oplus5V^{(1/2)}.  
\end{equation}
For any reduced subspace 
$W\subset V$ 
satisfying one of the following structural conditions:
\begin{align}
W &= V^{(5/2)}\oplus2V^{(3/2)}, \\
W &= 3V^{(3/2)}\oplus5V^{(1/2)}, \\
W &\supseteq4V^{(3/2)},
\end{align}
the total variance satisfies
\begin{equation}
\Delta^2(\mathfrak{su}_2)_{V/W}\geq B.
\end{equation}
where $B=0.2857, 0.091667$, and $0.2989$ respectively.
In other words, every state supported entirely on such reduced spaces obeys a strictly positive, state-independent lower bound on total collective-spin variance.
\end{theorem}

\medskip
The results indicate that the apparent negative minima are caused by interference between overlapping low-spin subspaces.
Once these degenerate components are removed, a stable, nonzero variance floor emerges.
We conjecture that the exact physical bound in the full quintipartite space remains strictly positive and will pursue an analytic derivation in future work.

\section{Discussion and Conclusion}

We have studied state-independent uncertainty relations for $\mathfrak{su}_2$ in multipartite qubit systems, using Clebsch--Gordan decomposition to connect representation structure with achievable variance bounds. The one-partite case yields a simple bound $\Delta^2(\mathfrak{su}_2) \ge j$ for spin-$j$ irreps. For reducible spaces, we derived a general
bound over direct sums, capturing interference between components.

\medskip
Our main findings are:
\begin{itemize}
    \item \textbf{Bipartite} ($n=2$): $(\mathbb{C}^2)^{\otimes 2}\cong V^{(1)}\oplus V^{(0)}$. The trivial subspace $V^{(0)}$ allows $\Delta^2=0$. On the reduced space $V^{(1)}$ one has $\Delta^2(\mathfrak{su}_2)\ge 1$.
    \item \textbf{Tripartite} ($n=3$): $(\mathbb{C}^2)^{\otimes 3}\cong V^{(3/2)}\oplus 2V^{(1/2)}$. No trivial sector occurs; our optimization yields the uniform, state-independent bound $\Delta^2(\mathfrak{su}_2)\ge 4/11$ for all states.
    \item \textbf{Quadripartite} ($n=4$): $(\mathbb{C}^2)^{\otimes 4}\cong V^{(2)}\oplus 3V^{(1)}\oplus 2V^{(0)}$. Although $\Delta^2$ can vanish on $V^{(0)}$, the reduced subspace $V^{(2)}\oplus 3V^{(1)}$ satisfies $\Delta^2(\mathfrak{su}_2)\ge 1/8$.
\item \textbf{Quintipartite ($n=5$):}
$(\mathbb{C}^2)^{\otimes5}\cong V^{(5/2)} \oplus4V^{(3/2)} \oplus5V^{(1/2)}$ contains ten irreducible sectors with multiplicities increasing for lower spins.
Applying Theorem~\ref{thm:variance-extremal-bound} yields an interior positive minimum $f_{\mathrm{int}}^{\min}\approx0.3031$, while boundary configurations involving multiple low-spin subspaces lead to an over-relaxed negative value.
By selectively removing three components, the minimal variance becomes strictly positive ($f_{\min}=0.2857$), indicating that destructive interference among degenerate subspaces is the source of the apparent negativity.
We therefore conjecture that the true physical bound in the full quintipartite space remains strictly positive, a result we plan to establish analytically in future work. {
Since variance is nonnegative by definition,
this indicates only that the relaxation is not sharp in the presence of large multiplicities.
We conjecture that the exact minimal total variance
in the full quintipartite space remains strictly positive.}

\end{itemize}

{
For a compact overview of the $n=1$--$5$ decompositions and bounds, see Table~\ref{tab:n-summary}.}

Beyond our structural results, it is instructive to compare them 
with representative variance-based and entropic uncertainty relations 
reported in the literature.
To place our results in context, Table~\ref{tab:URsummary} summarizes key variance- and entropic-based uncertainty relations reported in the literature, highlighting the dimensional scope and the availability of state-independent bounds. 
\begin{table}[h]
\centering
\scriptsize
\renewcommand{\arraystretch}{1.15}
\begin{tabularx}{\linewidth}{@{}l c c c  c c@{}}
\toprule
\textbf{Reference}  & \textbf{Single} & \textbf{Bipartite} & \textbf{Tripartite} & \textbf{Quadripartite} & \textbf{State-Indep.} \\
\midrule
\multicolumn{6}{c}{\textbf{Variance-based uncertainty relations}} \\
\midrule
Heisenberg~\cite{heisenberg1927} &  \ding{51} &  &  &  & \ding{55} \\[2pt]
Kennard~\cite{kennard1927} &  \ding{51} &  &  &  & \ding{51} \\[2pt]
Robertson~\cite{robertson1929} &  \ding{51} &  &  &  & \ding{55} \\[2pt]
Schr\"odinger~\cite{schrodinger1930} &  \ding{51} &  &  &  & \ding{55} \\[2pt]
Maccone--Pati~\cite{maccone2014} &  \ding{51} &  &  &  & \ding{55} \\[2pt]
Dammeier~et~al.~\cite{dammeier2015} & \ding{51} &  &  &  & \ding{51} \\[2pt]
de~Guise~et~al.~\cite{deGuise2018}  & \ding{51} &  &  &  & \ding{51} \\[2pt]
Schwonnek~et~al.~\cite{schwonnek2017}  & \ding{51} & \ding{51} &  &  & \ding{51} \\[2pt]
\textbf{This work}  & \ding{51} & \ding{51} & \ding{51} & \textbf{\ding{51}} & \ding{51} \\[2pt]
\midrule
\multicolumn{6}{c}{\textbf{Entropic uncertainty relations}} \\
\midrule
Maassen--Uffink~\cite{maassen1988} 
& \ding{51} &  &  &  & \ding{51} \\[2pt]

Renes--Boileau~\cite{Renes2009} 
&  & \ding{51} &  &  & \ding{51} \\[2pt]

Berta~et~al.~\cite{berta2010} 
&  & \ding{51} & \ding{51} &  & \ding{51} \\[2pt]

Haddadi~et~al.~\cite{haddadi2021}  
&  &  & \ding{51} &  & \ding{55} \\[2pt]

Zhang~et~al.~\cite{Zhang2023MultipartiteEUR} 
&  &  & \ding{51} & \ding{51} & \ding{55} \\[2pt]

Xu~et~al.~\cite{Xu2025ECT} 
&  &  & \ding{51} & \ding{51} & \ding{55} \\[2pt]

Wang~et~al.~\cite{Wang2024GEUR} 
&  &  & \ding{51} & \ding{51} & \ding{55} \\[2pt]
\bottomrule
\end{tabularx}

\caption{
Summary of representative variance-based and entropic uncertainty relations.  
“Single,” “Bipartite,” “Tripartite,” and “Quadripartite” indicate the number of subsystems explicitly analyzed.  
A check mark (\ding{51}) denotes the presence of an explicit analytical bound in that regime, while a cross (\ding{55}) indicates a state-dependent or heuristic formulation.
}
\label{tab:URsummary}
\end{table}

{
As shown in Table~\ref{tab:URsummary}, variance-based formulations 
have been rigorously established primarily for single systems and bipartite $\mathfrak{su}_2$ settings, 
with exact algebraic state-independent bounds currently established primarily in these regimes.
In contrast, entropic uncertainty relations have undergone substantial development, 
extending from bipartite memory-assisted formulations to generalized multipartite 
and conditional frameworks in recent years~\cite{berta2010,Zhang2023MultipartiteEUR,Xu2025ECT,Wang2024GEUR}. 
These entropic approaches provide tightened state-dependent bounds by incorporating conditional entropies and quantum memory terms, rather than relying on representation-theoretic variance analysis.}
{
Earlier variance-based works span a clear progression: 
Heisenberg’s heuristic formulation introduced the measurement–disturbance idea, 
followed by Kennard’s first rigorous inequality and Robertson’s general operator form. 
Schr\"odinger refined the relation by including covariance terms, 
and Maccone–Pati later strengthened the state-dependent version. 
Modern treatments by Dammeier, de~Guise, and Schwonnek established 
truly state-independent bounds for $\mathfrak{su}_2$ systems 
and connected them to entanglement detection. 
Despite these advances, variance-based state-independent bounds 
beyond the bipartite level have remained analytically inaccessible.}

\medskip
\noindent\textbf{{
Physical Interpretation and Implications:}} The representation-theoretic structure uncovered here admits several direct physical interpretations. 
First, the appearance of trivial irreducible components $V^{(0)}$ for even $n$ corresponds precisely to the existence of decoherence-free subspaces under fully collective $\mathrm{SU}(2)$ noise~\cite{lidar1998}. 
States supported entirely in these sectors exhibit vanishing collective variance and are invariant under global rotations, forming robust encodings against symmetric environmental perturbations. 
In contrast, odd-$n$ systems contain no trivial representation, and therefore no state can completely suppress collective spin fluctuations. 
This algebraic parity structure directly explains the emergence of strictly positive universal variance floors for odd subsystem numbers.

Second, the reduced-space lower bounds obtained here quantify fundamental limits on collective spin squeezing and metrological precision. 
Since the total variance of collective angular momentum determines achievable sensitivity in Ramsey-type interferometry~\cite{Giovannetti2006PRL,pezze2018}, our state-independent bounds provide intrinsic precision limits that depend only on the representation structure, independent of state preparation.

Third, the variance bounds derived for reducible spaces naturally serve as entanglement witnesses. 
Whenever a multipartite state violates the variance floor associated with separable spin sectors, entanglement must be present~\cite{hofmann2003,guhne2004}. 
Our direct-sum interference analysis clarifies how different irreducible components contribute constructively or destructively to collective fluctuations, thereby refining algebraic entanglement detection criteria.

Taken together, these observations show that the odd–even dichotomy revealed by the Clebsch–Gordan decomposition is not merely algebraic, but reflects  structural constraints on noise resilience, metrological enhancement, and multipartite quantum correlations in collective spin systems.

More generally, the Clebsch--Gordan multiplicity structure implies 
that the presence or absence of the trivial representation 
is completely determined by the parity of $n$ in the spin-$1/2$ case. 
This yields a systematic odd–even dichotomy in collective variance behavior, 
suggesting that representation multiplicities, rather than entanglement structure alone, 
govern the achievable uncertainty floors in multipartite spin systems.

\medskip
\noindent\textbf{{
Outlook:}}
The structural pattern uncovered here suggests a broader representation-driven program for multipartite uncertainty relations. 
For collective $\mathfrak{su}_2$ actions on $(\mathbb{C}^2)^{\otimes n}$, the Clebsch--Gordan decomposition reveals a parity-dependent mechanism: 
even subsystem numbers admit trivial irreducible components enabling variance suppression, 
whereas odd numbers necessarily enforce a strictly positive universal variance floor. 
We conjecture that this parity-driven structure persists for arbitrary $n$, 
and that a closed analytic expression for the minimal total variance 
may be derived directly from the multiplicities of irreducible spin sectors.

Beyond qubit systems, it is natural to extend the present representation-theoretic framework to higher-rank Lie algebras such as $\mathfrak{su}(d)$ 
and to collective actions on qudit ensembles. 
{
However, it appears that our current method can only provide bounds for some cases for higher spin $\mathfrak{su}(2)$-representations, thus new approach is needed for general higher spin representations} 
Moreover, establishing sharp, fully analytic state-independent variance bounds 
for general compact Lie algebras remains an open problem. 
Resolving this question would unify algebraic uncertainty principles 
with representation theory in a systematic way, 
providing intrinsic fluctuation limits for arbitrary multipartite quantum systems.

\bigskip
\centerline{\bf Acknowledgments}
The project is partially supported by the Simons Foundation.

\bigskip
\noindent{\bf Data availability statement} All data of the research have been included in the article. 

\bigskip
\bibliographystyle{amsalpha}

{
\vskip 1in
\section{Appendix}

We include the detailed derivation of the optimization used in Theorem
\ref{thm:variance-extremal-bound}.

\begin{lemma}\label{l:min} For positive parameters $j_k>0$ ($1\leq k\leq r$), let 
\begin{equation}
   f(x) = \sum_{k=1}^r j_k x_k - 2 \sum_{1 \le i < j \le r} j_i j_j x_i x_j , 
\end{equation}
be the real function of $x_1, \ldots, x_r$ defined on 
$[0, 1]^r$. Then the minimum of $f(x)$ under the constraint \( 
\sum_{k=1}^r x_k = 1 \) and $x_k>0$ is given by
\begin{equation}\label{e:mini}
  f_{\min} = \frac{r}{4r - 4} + \left( \frac{1}{2} + \frac{ \sum_{i=1}^r \frac{1}{j_i} }{4 - 4r} \right) \cdot \frac{ -2(r+1) + \sum_{i=1}^r \frac{1}{j_i} }{ \left( \sum_{i=1}^r \frac{1}{j_i} \right)^2 - (r - 1) \sum_{i=1}^r \frac{1}{j_i^2} }.  
\end{equation}
\end{lemma}
\begin{proof}
We introduce the Lagrangian:
\begin{equation}
  F(x, \lambda) = \sum_{i=1}^r j_i x_i + \sum_{i=1}^r j_i^2 x_i^2 - \left( \sum_{i=1}^r j_i x_i \right)^2 - \lambda \left( \sum_{i=1}^r x_i - 1 \right).  
\end{equation}
Taking partial derivatives with respect to \( x_i \), the stationarity conditions are:
\begin{equation}\label{eq:lagrange-critical}
\frac{\partial F}{\partial x_i} = j_i + 2 j_i^2 x_i - 2 j_i \sum_{k=1}^r j_k x_k - \lambda = 0.
\end{equation}
Multiplying \eqref{eq:lagrange-critical} by \( x_i \) and summing over \( i \), and using the normalization constraint \( \sum x_i = 1 \), we obtain:
\begin{equation}\label{eq:lambda-expression}
\lambda = \sum_{i=1}^r j_i x_i + 2 \sum_{i=1}^r j_i^2 x_i^2 - 2 \left( \sum_{i=1}^r j_i x_i \right)^2.
\end{equation}
Since the critical point corresponds to the solution of \eqref{eq:lagrange-critical}, the minimum function of $f(x)$ is given by:
\begin{equation}\label{eq:fmin-symmetric}
f_{\min} = \frac{1}{2} \lambda + \frac{1}{2} \sum_{i=1}^r j_i x_i.
\end{equation}
To eliminate $\lambda$ we divide \eqref{eq:lagrange-critical} by \( j_i \) and sum over \( i \), yielding:
\begin{equation}\label{eq:dotx-lambda}
\sum_{i=1}^r j_i x_i = \frac{ \lambda \sum_{i=1}^r \frac{1}{j_i} - r }{2(1 - r)}.
\end{equation}
Next,  we divide \eqref{eq:lagrange-critical} by \( j_i^2 \) and sum over \( i \) to obtain:
\begin{equation}\label{eq:togetlambda}
    \sum_{i=1}^r \frac{1}{j_i} + 2 - 2\sum_{i=1}^r j_i x_i  \sum_{i=1}^r \frac{1}{j_i} - \lambda \sum_{i=1}^r \frac{1}{j_i^2} = 0.
\end{equation}
Substituting \eqref{eq:dotx-lambda} into \eqref{eq:togetlambda}, we obtain an explicit expression for \( \lambda \) in terms of the \( j_i \) only:
\begin{equation}\label{eq:lambda-closedform}
\lambda = \frac{-2 + \frac{1}{r - 1} \sum_{i=1}^r \frac{1}{j_i}}{\frac{1}{r - 1} \left( \left( \sum_{i=1}^r \frac{1}{j_i} \right)^2 - (r - 1) \sum_{i=1}^r \frac{1}{j_i^2} \right)}.
\end{equation}

Finally, substituting equations \eqref{eq:dotx-lambda} and \eqref{eq:lambda-closedform} into \eqref{eq:fmin-symmetric}, we finally obtain a closed-form expression for the minimal value of the objective function given in \eqref{e:mini}.
\end{proof}
}

\bigskip
\bigskip
\end{document}